\documentclass[11pt]{article}

\usepackage{acl}

\usepackage{times}
\usepackage{latexsym}

\usepackage{xcolor}
\usepackage{lipsum}

\usepackage[T1]{fontenc}

\usepackage[utf8]{inputenc}

\usepackage{microtype}

\usepackage{graphicx}
\usepackage{adjustbox}
\usepackage{amsmath}
\usepackage{amsfonts}
\usepackage{algorithm}
\usepackage{bm}
\usepackage{soul}
\usepackage[noend]{algpseudocode}

\algrenewcommand\algorithmicindent{1.0em}
\algnewcommand{\LineComment}[1]{\State \(\triangleright\) #1}
\usepackage{amsthm}
\newtheorem{mydef}{Definition}[section]
\newtheorem{mythm}{Theorem}[section]

\title{DP-BART for Privatized Text Rewriting under Local Differential Privacy}

\author{Timour Igamberdiev \and Ivan Habernal	\\
	Trustworthy Human Language Technologies \\
	Department of Computer Science \\
	Technical University of Darmstadt \\
	\texttt{\{timour.igamberdiev, ivan.habernal\}@tu-darmstadt.de}\\
	\url{www.trusthlt.org}
}

\begin{document}
    \onecolumn
    \noindent \textbf{DP-BART for Privatized Text Rewriting under Local Differential Privacy}

    \medskip
    \noindent Timour Igamberdiev and Ivan Habernal

    \bigskip
    This is a \textbf{camera-ready version} of the article accepted for publication at the \emph{Findings of the Association for Computational Linguistics: ACL 2023}. The final official version will be published on the ACL Anthology website later in 2023: \url{https://aclanthology.org/}

    \medskip
    Please cite this pre-print version as follows.
    \medskip

\begin{verbatim}
@InProceedings{Igamberdiev.2023.ACL,
    title = {DP-BART for Privatized Text Rewriting
             under Local Differential Privacy},
    author = {Igamberdiev, Timour and Habernal, Ivan},
    publisher = {Association for Computational Linguistics},
    booktitle = {Findings of the Association for 
                 Computational Linguistics: ACL 2023},
    pages = {(to appear)},
    year = {2023},
    address = {Toronto, Canada}
}
\end{verbatim}
    \twocolumn

\maketitle
\begin{abstract}
Privatized text rewriting with local differential privacy (LDP) is a recent approach that enables sharing of sensitive textual documents while formally guaranteeing privacy protection to individuals.
However, existing systems face several issues, such as formal mathematical flaws, unrealistic privacy guarantees, privatization of only individual words, as well as a lack of transparency and reproducibility.
In this paper, we propose a new system `DP-BART' that largely outperforms existing LDP systems.
Our approach uses a novel clipping method, iterative pruning, and further training of internal representations which drastically reduces the amount of noise required for DP guarantees.
We run experiments on five textual datasets of varying sizes, rewriting them at different privacy guarantees and evaluating the rewritten texts on downstream text classification tasks.
Finally, we thoroughly discuss the privatized text rewriting approach and its limitations, including the problem of the strict text adjacency constraint in the LDP paradigm that leads to the high noise requirement.\footnote{Our code is available at \url{https://github.com/trusthlt/dp-bart-private-rewriting}.}
\end{abstract}

\section{Introduction}

\setstcolor{blue}

Protection of privacy is increasingly gaining attention in today's world, both among the general public and within the fields of machine learning and NLP.
One very common methodology for applying privacy to an algorithm is Differential Privacy (DP) \citep{Dwork.Roth.2013}.
In simple terms, DP provides a formal guarantee that any individual's contribution to a query applied on a dataset is bounded.
In other words, no individual can influence this query `too much'.

One particular method of applying DP to the domain of NLP is \textit{differentially private text rewriting}, in which an entire document is rewritten with DP guarantees by perturbing the original text representations.
For instance, given a document ``I would like to fly from Denver to Los Angeles this Thursday'', the system may rewrite it as ``Show me flights to cities in California this week''. If one is training a model on intent classification for airline travel inquiry systems, either document would be a useful data point.
In this way, we avoid using the original text that has uniquely identifiable qualities of a specific author, and instead create a privatized `synthetic' example.
This is in fact a form of local differential privacy (LDP), which is a stronger form of DP that is not limited to a specific dataset.

The benefits of an LDP text rewriting system are immense, where the output privatized dataset can be used for any downstream analysis.
We also avoid the problem of having to manually determine what specific tokens in a document are private, applying LDP to the entire document.
However, there is a significant difficulty in creating such a system, with a lot of perturbation required to achieve any reasonable privacy guarantees, leading to poor downstream utility.
In addition, there are several issues in existing DP text rewriting systems,
such as formal flaws having been discovered in their methodology \citep{habernal2021when}, older types of models used (e.g. single-layer LSTM, as in \citet{krishna2021adept}), 
high privacy budgets, 
as well as a lack of transparency in the claimed privacy guarantees, outlined in \citet{igamberdiev2022dp}.

To address these issues, we propose \textbf{DP-BART}, a DP text rewriting system under the local DP paradigm that improves upon existing baselines and consists of several techniques that can be directly applied to a pre-trained BART model \citep{lewis2020bart}, without having to design and train such a model from scratch.
Despite being a large transformer architecture, it can be easily used for data privatization, not requiring many resources.
Our methodology consists of a novel clipping method for the BART model's internal encoder representations, as well as a pruning and additional training mechanism that reduces the amount of DP noise that needs to be added to the data during the privatization process.

We summarize our contributions as follows.
First, we present our \textbf{DP-BART} model and its related methodologies, aimed at reducing DP noise and reaching a better privacy/utility trade-off.
For comparison, we use a reimplementation of the current primary baseline for this task, the ADePT model.
Second, we run experiments to investigate the privacy/utility trade-off of these models, using five unique datasets that gradually increase in size, evaluating rewritten texts on downstream text classification tasks.
Finally, we thoroughly examine the feasibility of the LDP text rewriting setting, investigating issues of the high noise requirement due to the strict text adjacency constraint, trade-offs between privacy and dataset size, what exactly is the object of privatization, required computational resources, as well as limitations of the approach as a whole and possible alternatives.

\section{Related Work}

We present a theoretical background on differential privacy, the BART model, and pruning for neural networks in Appendix~\ref{sec:appx-background}.

Applying differential privacy to neural network training and model publishing has converged to using a mainstream method, namely DP-SGD \citep{abadi2016deep}.
However, the task of text privatization is still broadly unexplored, with many unanswered questions remaining, such as dealing with the unstructured nature of text and explainability of the privacy guarantees provided to textual data \citep{klymenko-etal-2022-differential}.
\citet{mattern2022differentially} explored text rewriting with global differential privacy, sampling from a generative language model trained with DP.

There are only a few approaches that directly tackle the problem of differentially private text rewriting with LDP.
\citet{krishna2021adept} developed the ADePT system, which is an RNN-based text autoencoder that incorporates DP noise to its encoder output hidden state.
As described by \citet{habernal2021when}, ADePT had a formal error in calculating the Laplace noise scale, which resulted in it violating differential privacy.

A more recent text rewriting system is DP-VAE \citep{Weggenmann.et.al.2022.WWW}, which added constraints to the vanilla VAE model latent space \citep{kingma2014auto} to obtain a bounded sensitivity on its mean and variance parameters. Despite the high difficulties of the task, the paper reports surprisingly high performance for high privacy standards.
Since their experimental description lacks some key details and the code base is not public, we cannot reproduce their approach.

In addition, there are a number of word-level DP systems \citep{feyisetan2019leveraging,xu-etal-2020-differentially,bo2021er}, where individual word embeddings are perturbed with DP, with new words then sampled close to these privatized vectors.
As \citet{Mattern.et.al.2022.arXiv} point out, there are several shortcomings of such approaches, including a lack of obfuscating syntactic information and the inability to provide proper anonymization.
In essence, these methods do not privatize a full utterance, but only single words.

\section{Methods}
\label{sec:methods}

We outline this section as follows. First, we briefly describe the baseline method we use, being a modified version of the ADePT system by \citet{krishna2021adept}.
Next, we investigate two main issues with applying a local DP system such as ADePT to a transformer model, namely extreme sensitivity and computational infeasibility, described in Sections~\ref{sec:ldp-transformers-extreme-sensitivity} and \ref{sec:high-computational-demand}, respectively.

We then demonstrate several novel mechanisms which tackle these issues and provide numerous benefits in the privacy/utility trade-off for the local DP setting.
Section~\ref{sec:dp-bart-clv} describes the clipping by value module, with an additional analysis on determining optimal settings for it provided in Appendix~\ref{sec:appx-c-value-selection}.
Sections~\ref{sec:dp-bart-pr} and \ref{sec:dp-bart-pr-plus} then describe the neuron-based pruning methods which significantly reduce the amount of noise that needs to be added to the model for a given privacy budget and increase model robustness to noise through further noisy training. Low-level specifics on the pruning methods are further provided in Appendix~\ref{sec:appx-pruning}.

\subsection{Baseline (ADePT)}

ADePT starts out with a standard autoencoder architecture. Given an input document $x$, an encoder function $\textsc{Enc}$ calculates a latent vector representation $z$. This representation is then sent to a decoder function $\textsc{Dec}$, which reconstructs the original text $\hat{y}$.
ADePT uses a single-layer, unidirectional LSTM for both the encoder and decoder.
\begin{equation}
    z = \textsc{Enc}(x) \quad \text{and} \quad \hat{y} = \textsc{Dec}(z)
    \label{eqn:basic-autoencoder}
\end{equation}

To incorporate differential privacy into this model, the unbounded latent vector $z \in \mathbb{R}^{n}$ (where $n$ is the size of the autoencoder's hidden dimension) is bounded by its norm and the clipping constant $C \in \mathbb{R}$.
Laplace or Gaussian noise ($\eta$) is then added to the resulting vector, from which the decoder reconstructs the original sequence, $\hat{y}$. For comparison with our primary methodologies below, we refer to this as the clipping by norm module, outlined in equation~\ref{eqn:cln-module}.
\begin{equation}
    z' = z \cdot \min \left(1, \frac{C}{||z||_2} \right) + \eta
    \label{eqn:cln-module}
\end{equation}

In our experiments, we make three adjustments to this system. First, we fix a theoretical issue in the sensitivity calculation for equation~\ref{eqn:cln-module}, outlined in \citet{habernal2021when}. Instead of using the sensitivity of $2C$ for the Laplace noise scale, outlined in Theorem 1 of \citet{krishna2021adept}, we instead use the corrected sensitivity of $2C\sqrt{n}$ from Theorem 5.1 of \citet{habernal2021when}.
Second, the `classical' Gaussian mechanism guarantees privacy only for $\varepsilon < 1$ \citep[p.~262]{Dwork.Roth.2013}. We therefore utilize the Analytic Gaussian mechanism \citep{balle2018improving} instead, which allows us to use $\varepsilon \geq 1$.
Finally, we fix an issue with the pre-training procedure of the model.
In \citet{krishna2021adept}, ADePT was pre-trained on the downstream datasets with clipping, but without the added DP noise from equation~\ref{eqn:cln-module}. \citet{igamberdiev2022dp} demonstrated that this results in significant memorization by the model of the input documents, even after adding DP noise during the rewriting process. In order to remedy this, we therefore pre-train the autoencoder model on a public corpus, unrelated to the downstream datasets.

\subsection{Applying LDP to Transformers}
\label{sec:applying-ldp}

There are two main issues in applying a transformer model to a local DP setting similar to ADePT, outlined below.

\subsubsection{Using LDP in pre-trained transformers suffers from extreme sensitivity}
\label{sec:ldp-transformers-extreme-sensitivity}
First, we need a significantly larger amount of noise to be added to the model, due to the increased size of the encoder output vector. Due to the cross-attention mechanism typical of transformer models, the full output vector for the BART encoder is of size $d_{tok} \times l$, where $d_{tok}$ is the hidden size for a particular token, while $l$ is the sequence length. 
For the smaller \texttt{bart-base} model, 
using a short sequence length of $20$, this results in a dimensionality of $768 \times 20 = 15360$. 
In comparison, ADePT's encoder output vector dimensionality is only $1024$ in our configuration.

\subsubsection{High requirement of computational resources for pre-training}
\label{sec:high-computational-demand}

We experimented with clipping by norm for BART, similarly to ADePT, but found that it destroys any useful representations of the model (even prior to adding the DP noise).
Additional pre-training of BART that would incorporate clipping by norm turned out to be ineffective.

The remaining option to learn a model with clipping by norm would be to pre-train the model from scratch.
Unlike the small ADePT model, which is a unidirectional, single-layer LSTM, pre-training a BART transformer from scratch is computationally  infeasible on an academic budget.
While the details of BART's computational requirements are not described in \citet{lewis2020bart}, we can estimate this for the relatively small \texttt{bart-base} model of 139M parameters that was released by the original authors,\footnote{\url{https://github.com/facebookresearch/fairseq/tree/main/examples/bart}} by comparison with other similar-sized models.
For instance, the BERT model \citep{devlin-etal-2019-bert}, with less parameters (110M for \texttt{bert-base}), was pre-trained for 4 days on up to 16 TPUs, as described on the authors' Github repository.\footnote{\url{https://github.com/google-research/bert}}

\subsection{DP-BART-CLV (Clipping by Value)}
\label{sec:dp-bart-clv}

\begin{figure}[h]
    \centering
    \includegraphics[width=\linewidth]{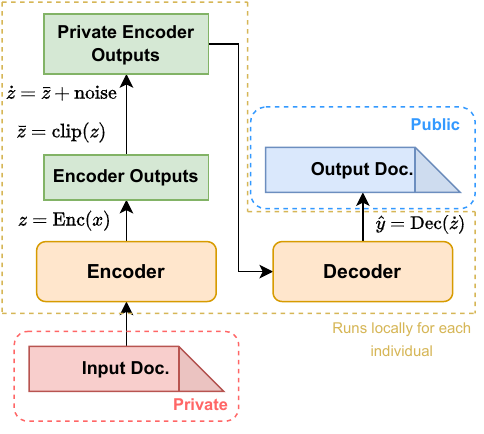}
    \caption{DP-BART-CLV}
    \label{fig:dp-bart-clv}
\end{figure}

To address the issues with clipping by norm, we developed the \textbf{DP-BART-CLV} model, shown in Figure~\ref{fig:dp-bart-clv}.
We analyzed the internal representations of a pre-trained BART model's encoder output vector values, using a public dataset. We found that these are mostly bounded within a couple of standard deviations from their mean. We present this analysis in detail in Appendix~\ref{sec:appx-c-value-selection}.

To avoid significantly altering these representations, we can therefore use clipping by value (CLV), as in equation~\ref{eqn:clv-clipping-per-dim}.
\begin{equation}
    \bar{z}_i = \min(\max(z_i, C_{min}), C_{max})
    \label{eqn:clv-clipping-per-dim}
\end{equation}
for any dimension $i$ in the encoder output vector $z$, a set minimum threshold $C_{min}$ and maximum threshold $C_{max}$.
The bulk of values centered around the mean of $z$ are thus left the same, without being rescaled as in equation~\ref{eqn:cln-module}.
Since these values were also found to be symmetrically distributed, we modify equation~\ref{eqn:clv-clipping-per-dim} to set $C = C_{max} = -C_{min}$, as in equation~\ref{eqn:clv-clipping-per-dim-modified}.
\begin{equation}
    \bar{z}_i = \min(\max(z_i, -C), C)
    \label{eqn:clv-clipping-per-dim-modified}
\end{equation}

The pipeline for \textbf{DP-BART-CLV} is as follows.
We first initialize a BART model using a pre-trained checkpoint, where pre-training was again done on a public dataset, separate from the downstream datasets that are to be privatized.

For a given document, we put it through the encoder of the model at inference time, obtaining the encoder output vector $z$, as in equation~\ref{eqn:clv-encoder}.
\begin{equation}
    z = \textsc{Enc}(x)
    \label{eqn:clv-encoder}
\end{equation}
where $x$ is the input sequence and $\textsc{Enc}$ is the encoder of the BART model.
While the BART model outputs the encoder's last hidden state as $z \in \mathbb{R}^{l \times d_{tok}}$ for each mini-batch, we flatten this vector to be $z \in \mathbb{R}^n$, where $n = l \cdot d_{tok}$.
Clipping is then performed as in equation~\ref{eqn:clv-clipping-full},
\begin{equation}
    \bar{z} = \textsc{clip}(z)
    \label{eqn:clv-clipping-full}
\end{equation}
where \textsc{clip} is carried out for every dimension of the vector, according to equation~\ref{eqn:clv-clipping-per-dim-modified}.

With this clipping mechanism in place, we can now calculate its sensitivity, in order to determine the scale of noise to add in the DP setting. This is outlined in Theorems~\ref{thm:dp-bart-clv-laplace} and \ref{thm:dp-bart-clv-gaussian} below.
\begin{mythm}
\label{thm:dp-bart-clv-laplace}
Let $f: \mathbb{R}^n \rightarrow \mathbb{R}^n$ be a function as in equation~\ref{eqn:clv-clipping-full}. The $\ell_1$ sensitivity $\Delta_1 f$ of this function is calculated as in equation~\ref{eqn:laplace-clv-sens},
where $C \in \mathbb{R}: C > 0$ is the clipping constant and $n \in \mathbb{N}$ is the dimensionality of the vector.
\end{mythm}

\begin{equation}
    \Delta_1f = 2Cn
    \label{eqn:laplace-clv-sens}
\end{equation}

\begin{proof}
See Appendix~\ref{sec:appx:thm-clv-laplace-proof}.
\end{proof}

\begin{mythm}
\label{thm:dp-bart-clv-gaussian}
Let $f: \mathbb{R}^n \rightarrow \mathbb{R}^n$ be a function as in equation~\ref{eqn:clv-clipping-full}. The $\ell_2$ sensitivity $\Delta_2 f$ of this function is calculated as in equation~\ref{eqn:gaussian-clv-sens}, where $C \in \mathbb{R}: C > 0$ is the clipping constant and $n \in \mathbb{N}$ is the dimensionality of the vector.
\end{mythm}

\begin{equation}
    \Delta_2f = 2C\sqrt{n}
    \label{eqn:gaussian-clv-sens}
\end{equation}

\begin{proof}
See Appendix~\ref{sec:appx:thm-clv-gaussian-proof}.
\end{proof}

We then add noise to this clipped vector,
as in equation~\ref{eqn:clv-mechanism}.
\begin{equation}
    \dot{z} = \bar{z} + (Y_1, \dots, Y_n)
    \label{eqn:clv-mechanism}
\end{equation}
where each $Y_i$ is drawn i.i.d.\ from $\text{Lap}(\frac{\Delta_1}{\varepsilon})$ for the Laplace mechanism \citep{Dwork.Roth.2013} or $\mathcal{N}(0, (\frac{\alpha \Delta_2}{\sqrt{2\varepsilon}})^2)$ for the Analytic Gaussian mechanism, where $\alpha$ is calculated according to Algorithm 1 of \citet{balle2018improving}.

Decoding is then performed auto-regressively (e.g. using beam search), as usual, using this perturbed $\dot{z}$ encoder output vector, instead of the original $z$ vector, as in equation~\ref{eqn:clv-decoder}.
\begin{equation}
    \hat{y} = \textsc{Dec}(\dot{z})
    \label{eqn:clv-decoder}
\end{equation}
where $\hat{y}$ is the model's output prediction of the reconstructed input sequence $x$.
By standard arguments, the \textbf{DP-BART-CLV} model satisfies $(\varepsilon, 0)$-DP for the Laplace mechanism and $(\varepsilon, \delta)$-DP for the Analytic Gaussian mechanism, as outlined in equation~\ref{eqn:clv-mechanism} \citep{Dwork.Roth.2013,balle2018improving}.

\subsection{DP-BART-PR (Pruning)}
\label{sec:dp-bart-pr}

We develop the \textbf{DP-BART-PR} model in order to address the remaining issue of dimensionality, outlined in Section~\ref{sec:ldp-transformers-extreme-sensitivity}.
The \textbf{DP-BART-CLV} model, while being resource-efficient, still has the issue of a large dimensionality for the encoder output vectors, since in equations~\ref{eqn:laplace-clv-sens} and \ref{eqn:gaussian-clv-sens}, the sensitivity is multiplied by a factor of $n$ and $\sqrt{n}$, respectively, which in turn results in a larger noise scale.

\textbf{DP-BART-PR}, addressing both the resource and dimensionality issues, is an extension to the above \textbf{DP-BART-CLV}, with an additional iterative pruning/training mechanism applied to it.
The procedure is outlined in Figure~\ref{fig:dp-bart-pr} and Algorithm~\ref{alg:dp-bart-pr} of Appendix~\ref{sec:appx-algo}.

\begin{figure}[!ht]
    \centering
    \includegraphics[width=\linewidth]{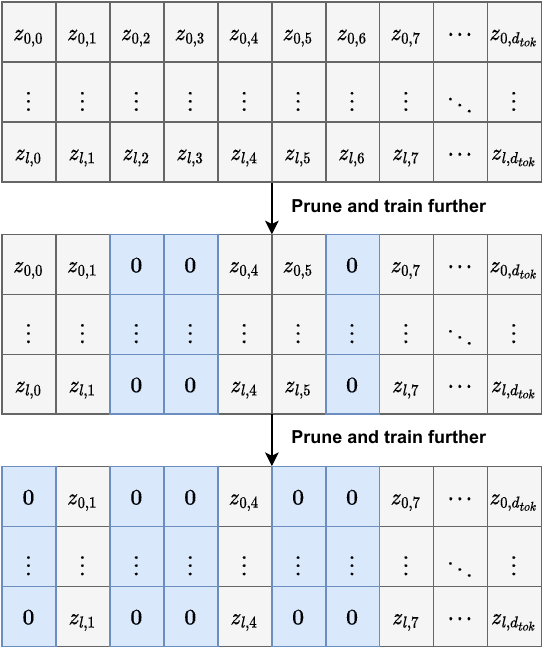}
    \caption{Pruning and re-training procedure for the DP-BART-PR model, illustrated for one document. Each $i^{th}$ neuron from a set of indices is set to $0$ for all tokens of the encoder output vectors $z \in \mathbb{R}^{l \times d_{tok}}$. These neuron indices are the same for any document. This process is repeated iteratively until performance starts to degrade.}
    \label{fig:dp-bart-pr}
\end{figure}

As for \textbf{DP-BART-CLV}, we first load a pre-trained BART model checkpoint. Each input token will have an encoder output representation of dimensionality $d_{tok}$. For every token in the sequence, we prune a certain percentage of these neurons by setting them to 0. Importantly, \textit{these pruned neurons are the same for every single input document}.
The criteria for selecting these pruned neurons is discussed in more detail in Appendix~\ref{sec:appx-pruning}.

Following this pruning step, we train the model for $k$ iterations to compensate for possible lost performance from pruning.
This step is performed on an external public dataset, unrelated to any downstream texts that are to be privatized.
During this process, we also clip each dimension of the BART encoder output vector $z_i$ according to equation~\ref{eqn:clv-clipping-per-dim-modified}, to encourage representations to be constrained within the ranges $-C$ and $C$ to reduce potential negative performance impacts of clipping during the rewriting phase.

We note that only a few data points are necessary for this additional training step, maintaining the low-resource setting, outlined in Appendix~\ref{sec:appx-hyperparameters}.
We then continue this two-step process iteratively, until a desired dimensionality reduction of the encoder output vector is reached.
At the end of this process, the resulting model weights are frozen and the final pruned indices of the encoder output vector $z$ are saved.
The model is then used for text rewriting at inference time, just like in \textbf{DP-BART-CLV}, but with the additional pruning step, using the saved indices.

As a result of this process, we can significantly reduce $n$ in Equations~\ref{eqn:laplace-clv-sens} and \ref{eqn:gaussian-clv-sens}, which in turn reduces the resulting noise scale used in equation~\ref{eqn:clv-mechanism}. With less noise added to the encoder output vectors for any given $\varepsilon$ value, we can thus expect a better privacy/utility trade-off.

This pruning procedure can thus be seen as a \textit{privacy/utility tuning knob}. With more pruning, we reduce the size of $n$, therefore requiring less added noise for a given $\varepsilon$ value in the DP setting. At the same time, more pruning reduces the model's expressivity with less dimensions, which will result in an inevitable performance drop after reaching a certain pruning threshold.
We noticed that pruning a few dimensions (e.g.\ 25\% of neurons) can recover basically all of the performance of the model with some additional training steps, but after a certain point this starts to degrade. The `sweet spot' we found is at approximately 75\% of neurons. Additional discussions on these points can be found in Appendix~\ref{sec:appx-pruning}.
We would like to stress again that these pruning adjustments are made just once and using public data only, after which the final model can be used locally by any individual for their own data privatization.

\subsubsection{Proof that DP-BART-PR is differentially private}

\begin{mythm}
\label{thm:dp-bart-pr-proof}
The \textbf{DP-BART-PR} model, combining Algorithm~\ref{alg:dp-bart-pr} and the above \textbf{DP-BART-CLV} procedure, summarized in equation~\ref{eqn:clv-mechanism}, satisfies $(\varepsilon, 0)$-DP when using the Laplace mechanism and $(\varepsilon, \delta)$-DP when using the Analytic Gaussian mechanism.
\end{mythm}

\begin{proof}
See Appendix~\ref{sec:appx-thm-pr-proof}.
\end{proof}

\subsection{DP-BART-PR+}
\label{sec:dp-bart-pr-plus}

We further augment the above \textbf{DP-BART-PR} model by incorporating additional training steps with added DP noise. 
This model follows the same procedure for iterative pruning and additional training, as outlined in algorithm~\ref{alg:dp-bart-pr}, but we add further training iterations on the pruned model with added DP noise to the clipped encoder output representations, as in equation~\ref{eqn:clv-mechanism}.
For example, using the Analytic Gaussian mechanism at $\varepsilon=500$, at each iteration we clip the encoder output vectors $z$ from equation~\ref{eqn:clv-encoder} and add the appropriate amount of Gaussian noise based on the sensitivity from equation~\ref{eqn:gaussian-clv-sens}.

The idea behind this additional training is to help the model to better decode from the noisified encoder representations.
As with \textbf{DP-BART-PR}, for \textbf{DP-BART-PR+} we perform these additional training iterations on a public dataset, unrelated to the downstream datasets for privatized text rewriting.
A separate model is prepared for each individual privacy budget $\varepsilon$.

\section{Experiments}

\subsection{Datasets}

We perform experiments on five English-language textual datasets, each gradually increasing in size (Table~\ref{tab:dataset-stats}).
For comparison with \citet{krishna2021adept}, we use ATIS \citep{dahl1994expanding} and Snips \citep{coucke2018snips} as our `small' datasets, with the task of multi-class intent classification. We use the same train/validation/test split as in \citet{goo2018slot}.
For a medium-sized dataset, we use the popular IMDb dataset \citep{maas2011learning}, on the binary classification task of movie review sentiment analysis. For this, as well as the following two datasets, we use a validation partition by randomly selecting 20\% of the training set.

For a large dataset, we use the dataset from \citet{grasser2018aspect}, which is a collection of drug reviews from the website Drugs.com, also with the task of binary sentiment analysis as in \citet{shiju2022}.
This dataset, although publicly available, 
closely simulates a sensitive dataset in need of privacy protection,
with detailed descriptions by users of their medical conditions and experiences with different treatments.

Our final dataset is the much larger Amazon Customer Reviews dataset \citep{he2016ups}, of which we take a 2M subset of reviews from various categories (e.g. electronics, office products), from the full 144M. As with Drugs.com, we modify the original five-star sentiment score to a binary classification task, with four or more stars being the `positive' class, while the rest are `negative'.
We refer to Appendix~\ref{sec:appx-dataset-prep} for more details.

\begin{table}[h]
    \centering
    \begin{tabular}{lr|rr}
		\textbf{Dataset}	& \textbf{Classes}	& \textbf{\texttt{\#} Trn.+Vld.}	&\textbf{\texttt{\#} Test}		\\ \hline
		ATIS	&26	&4,978	&893	\\
		Snips	&7	&13,774	&700		\\
		IMDb	&2	&25,000	&25,000	\\
		Drugs.com	&2 	&161,297 	&53,766 	\\
		Amazon	&2	&1,904,197	&211,605	\\
    \end{tabular}
    \caption{Dataset statistics. Trn.: Train, Vld.: Validation. Size represents number of documents.}
    \label{tab:dataset-stats}
\end{table}

\subsection{Experimental Setup}
\label{sec:exp-setup}

We have three main experimental configurations.
The first is the \textbf{original} setting, where we run experiments on our downstream datasets without any rewriting or DP.
The second configuration is \textbf{rewrite-no-dp}, where we utilize each of the four models outlined in Section~\ref{sec:methods} at $\varepsilon=\infty$ (\textbf{ADePT}, \textbf{DP-BART-CLV}, \textbf{DP-BART-PR}, \textbf{DP-BART-PR+}).
Finally, the third and main configuration is \textbf{rewrite-dp}, where we compare the above four models, this time at various privacy settings 
($\varepsilon \in [10, 2500]$, Laplace and Analytic Gaussian mechanisms).

For \textbf{rewrite-no-dp} and \textbf{rewrite-dp}, our experimental pipeline consists of the following four steps, depending on the specific model used:

\noindent \hangindent=0.7cm \textbf{Pre-training:} The model is pre-trained on a large public corpus. For ADePT, we use 50\% of the Openwebtext corpus \citep{Gokaslan2019OpenWeb}.
For all our BART experiments, we load a pre-trained \texttt{facebook/bart-base} model.\footnote{Available from \url{https://huggingface.co/facebook/bart-base}}

\noindent \hangindent=0.7cm \textbf{Further training:} Only for DP-BART-PR and DP-BART-PR+, again performed using the Openwebtext corpus. It helps the model adjust to pruning and DP noise, respectively (as outlined in Sections~\ref{sec:dp-bart-pr} and \ref{sec:dp-bart-pr-plus}).
More details on the amount of further training in Appendix~\ref{sec:appx-hyperparameters}.

\noindent \hangindent=0.7cm \textbf{Rewriting:} We take a pre-trained model and rewrite one of the downstream datasets.

\noindent \hangindent=0.7cm \textbf{Downstream:} We take the rewritten dataset (training and validation partitions) and run downstream experiments on it using a pre-trained BERT model with a classification head on top. We use the rewritten validation set for hyperparameter optimization (see Appendix~\ref{sec:appx-hyperparameters}) and the original test set for final evaluations. See Appendix~\ref{sec:appx-downstream-setup} for details on the downstream model.

\noindent In the \textbf{original} setting, we use the same downstream model as above, using the original datasets instead of the rewritten ones.

\paragraph{Evaluation}

We perform two types of evaluations for the above experimental settings: intrinsic and extrinsic.
For our extrinsic evaluation we measure the test $F_1$ scores on the downstream task performance. This is the primary utility metric of the rewritten texts, with privacy correspondingly quantified with the $\varepsilon$ value.
We expect that even if a text may be rewritten to look very different from the original input, it could still have enough downstream task-specific information remaining to properly train a model on this task (e.g. the sentiment of a document in the case of sentiment analysis).
This is in fact the `sweet spot' we are looking for, removing identifying elements of the author, but still retaining some key features from the input for good downstream performance.

We also measure BLEU scores \citep{papineni2002bleu} for our intrinsic evaluation, discussed in more detail in Appendix~\ref{sec:appx-intrinsic-eval-and-detailed-res}.

\section{Results}
\label{sec:results}

\begin{figure*}[!ht]
    \centering
    \includegraphics[width=\linewidth]{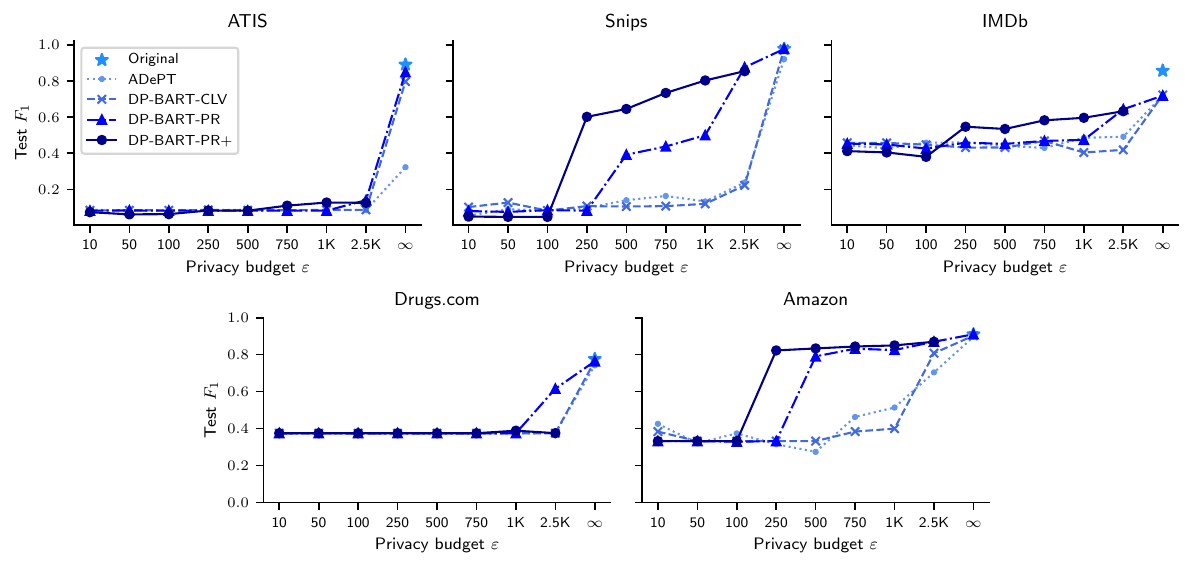}
    \caption{Downstream test $F_1$ results (macro-averaged) for each dataset, using the four model types. Lower $\varepsilon$ corresponds to better privacy. Both \textbf{original} and \textbf{rewrite-no-dp} results can be seen on the right of each graph at $\varepsilon=\infty$. The rest of the results represent the \textbf{rewrite-dp} setting at different $\varepsilon$ values.}
    \label{fig:results}
\end{figure*}

Figure~\ref{fig:results} shows our downstream test $F_1$ results for all datasets, at varying values of $\varepsilon$.
We report results for the Analytic Gaussian mechanism, which nearly always outperformed those of the Laplace mechanism.
We present results in tabular form with mean and standard deviations in Appendix~\ref{sec:appx-intrinsic-eval-and-detailed-res}.
Additionally, we present sample rewritten texts in Appendix~\ref{sec:appx-sample-texts}.
We outline the main patterns as follows.

\paragraph{DP-BART-PR+ performs best}
DP-BART-PR+ reaches the best privacy/utility trade-off for the majority of datasets, having the highest scores at the lower $\varepsilon$ values.
DP-BART-PR results are second-best for most datasets, performing better than DP-BART-CLV and ADePT, which are low for the majority of configurations.
The overall results hierarchy can be clearly seen in the Snips dataset, where at $\varepsilon=500$, DP-BART-PR+ reaches $F_1$ $0.65$, DP-BART-PR at $0.39$, while both DP-BART-CLV and ADePT are below $F_1$ $0.15$.

\paragraph{Original vs. rewritten}
Results for the \textbf{original} setting are generally on-par with those of the \textbf{rewrite-no-dp} setting. For instance, Snips original $F_1$ is $0.98$, and $\varepsilon=\infty$ with rewriting is also at $F_1$ of $0.98$ for DP-BART-PR, being very similar for the other three models.
One exception to this is IMDb, which has a drop from original $F_1$ $0.86$ to $0.72$ for all models.
This can be explained by the fact that the \textbf{original} settings use longer sequence lengths, while both \textbf{rewrite-no-dp} and \textbf{rewrite-dp} settings are limited to a sequence length of $20$.
This is not a problem for datasets such as ATIS and Snips, since their documents are generally very short, mostly limited to brief user inquiries.
For a dataset such as IMDb, however, which consists of detailed reviews by individuals, limiting the sequence length results in a loss of valuable information.

\paragraph{Epsilon vs. dataset size}
Regardless of dataset size, we can see a drop in results for all models as $\varepsilon$ is decreased.
With the models incorporating pruning, this drop appears at later $\varepsilon$ values, such as DP-BART-PR+ on the Amazon dataset moving down from $F_1$ $0.82$ at $\varepsilon=250$ to $F_1$ $0.33$ at $\varepsilon=100$, and DP-BART-PR from $F_1$ $0.79$ at $\varepsilon=500$ to $F_1$ $0.33$ at $\varepsilon=250$.
A similar pattern can be seen for the Snips dataset, despite being far smaller than Amazon, while the Drugs.com dataset shows low results throughout, for all model types.
The smallest dataset, ATIS, also performs poorly, which can be explained by the large number of classes and few data points for learning the task in the noisy setting.
We can generally see that a larger dataset size does not necessarily mean better results at lower $\varepsilon$ values, although the significantly larger Amazon dataset does show the best results.

\section{Discussion}

\paragraph{Reducing noise for text rewriting with LDP}

We have shown that it is possible to reduce the amount of noise in the LDP setting of privatized rewriting, in order to obtain more useful rewritten texts for downstream tasks.
To compare DP-BART-CLV vs. DP-BART-PR, we can examine the resulting $\ell_2$ sensitivity from equation~\ref{eqn:gaussian-clv-sens} ($\Delta_2f = 2C\sqrt{n}$).
Setting sequence length $l=20$ and $C=0.1$, as in our experiments, without pruning we have a dimensionality of $n=768 \cdot 20=15360$, hence $\Delta_2f = 2 \cdot 0.1 \cdot \sqrt{15360} \approx 24.79$. With pruning we are able to remove 76.30\% of those $n$ neurons, with only $n=182 \cdot 20=3640$ remaining. The $\ell_2$ sensitivity thus becomes $\Delta_2f = 2 \cdot 0.1 \cdot \sqrt{3640} \approx 12.07$.

Plugging this into the Analytic Gaussian mechanism's noise scale calculation from \citet{balle2018improving}, with $\delta=10^{-5}$ and $\varepsilon=500$, we have $\sigma^2 = 0.8958$ without pruning and $\sigma^2 = 0.4362$ with pruning.
We can therefore see that, \textbf{with DP-BART-PR, we are able to reduce the noise scale by more than half}.

\paragraph{Pre-training and computational resources}

Ultimately, a very effective way to prepare a model for privatized text rewriting would be to pre-train it from scratch, being fully in control of hyperparameters such as the dimensionality $n$ of the encoder output vectors $z$, which determines the $\ell_1$ and $\ell_2$ sensitivities from equations~\ref{eqn:laplace-clv-sens} and \ref{eqn:gaussian-clv-sens}, respectively.
In addition, the whole model could be pre-trained with added noise and clipping mechanisms, potentially being even more robust than our approach in DP-BART-PR+, where we incorporate further noisy training.
We noticed for DP-BART-PR+ that the lower the $\varepsilon$ value we use, the more additional training iterations the model needs to properly reduce the validation loss.

This demonstrates that, also in the setting of pre-training from scratch, we would need to train for more iterations in order to reach lower $\varepsilon$ values.
This can pose serious challenges, however, for reasons of computational demand discussed in Section~\ref{sec:high-computational-demand}.
DP-BART-PR+ can therefore be seen as a sweet spot approach, where we only need a few additional training iterations and can still achieve a significant dimensionality reduction through pruning, as well as additional robustness to noise.

\paragraph{What is being privatized}

It is very important to be clear on exactly what information is being privatized when performing text rewriting with LDP.
Since we are working with DP at the document level, the entire document is a `data point', hence any choice and combination of words for a given sequence would be a unique identifier.
We thus avoid the problem of having to choose what specific tokens are `private' within the document. This is crucial, since stylistic aspects of an author can be very abstract, with subtle syntactic and vocabulary choices.

Another significant benefit of such an approach, is that we are not limiting ourselves to any specific downstream analysis (e.g.\ sentiment of a document), being \textit{task agnostic}.
However, this also means that, for any given document, \textit{any other document is neighboring}, since we are in the LDP setting.
This leads us to a serious discussion on the limitations of such an approach in Section~\ref{sec:limitations}.

An additional question arises of whether one dataset may have multiple documents associated with one individual.
There are several ways to go about dealing with this.
One standard approach in differential privacy is to linearly scale the $\varepsilon$ parameter.
Thus, if there are $k$ documents associated with a given individual, then a privacy budget of $k\varepsilon$ is accounted in total \citep{Dwork.Roth.2013}.
Another option would be to simply append all texts associated with one individual into a single `document', rewriting this using just a single $\varepsilon$ privacy budget.

\section{Conclusion}

We have proposed DP-BART, a novel methodology for LDP-based privatized text rewriting, which outperforms existing methods. We have demonstrated our method's privacy/utility trade-off, the relations between the privacy budget and dataset size, and discussed limitations of the privatized text rewriting approach as a whole. Future research directions include utilizing large-scale pre-training to potentially reach a better privacy/utility trade-off, as well as investigating domain specific text rewriting for relaxing the strict requirements of the LDP approach.

\section{Limitations}
\label{sec:limitations}

\paragraph{Domain of public training texts}

In preparing the DP-BART models, it is important to take into account the domain of the public data that is used to (1) pre-train the original BART model, and (2) perform additional training iterations (DP-BART-PR and DP-BART-PR+).
This will ultimately have an impact on the model's effectiveness for text privatization, depending on the nature of the downstream texts.
For example, if this training data is restricted to news articles, then there may be limited performance for rewriting texts that are further from this domain, such as internet comments.
Another obvious limitation is the language of the public data. If the model is trained on a monolingual English corpus, then it would not be possible to use it for rewriting texts from other languages.

The public data used for our experiments consists of news, web text, stories and books \citep{lewis2020bart, Gokaslan2019OpenWeb}.
We expect that expanding this to include more data and more varied domains will lead to better performance in a greater diversity of texts and downstream tasks.

\paragraph{LDP for text rewriting}

For every output document, any two inputs, no matter how similar or distinct, are considered neighboring.
If we have a small sequence length of $20$ tokens, with a relatively small vocabulary of $1000$ words, then the total number of possible combinations is $1000^{20}$, which is $10^{60}$!
While we compress these documents into a latent vector with a limited range and dimensionality, the strict adjacency constraints are still present.
We can therefore expect an inevitable utility drop when using more reasonable $\varepsilon$ values (e.g. $\varepsilon=1$).

With more sophisticated architectures, we have shown that it is possible to push this $\varepsilon$ value down to some extent. However, our lowest $\varepsilon$ is still too high to carry over into real-world applications of privacy preservation. As outlined by \citet{hsu2014differential}, values of $\varepsilon$ for different applications in the DP literature can range from 0.01 to 10. Choosing the right $\varepsilon$ value depends on the specific queries that are computed and the nature of the data \citep{lee2011much}.

For our case, the value of $\varepsilon$ can be interpreted in the following manner.
The $\varepsilon$-LDP mechanism that we are applying to our data makes any two input texts rewritten to be indistinguishable up to a factor of $e^{\varepsilon}$.
More formally, \textit{for any two input texts} $x$ and $y$ to our LDP model $\mathcal{M}$:
\begin{equation}
    \frac{\text{Pr}[\mathcal{M}(x) = z]}{\text{Pr}[\mathcal{M}(y) = z]} \leq e^{\varepsilon},
\end{equation}
where $z$ is a given output text rewritten by the model.

This means that, when we set $\varepsilon = 250$, then any two texts will remain indistinguishable up to a factor of $e^{250}$.
This is a very weak bound and, while it could provide some empirical privacy guarantees, on a theoretical level the privacy protection is not very strong.
We can also see how this bound becomes exponentially stronger, as we decrease $\varepsilon$.

It may therefore make sense to take a slightly less strict approach to text adjacency, for instance moving into \textit{domain specific} text rewriting. For example, text rewriting could be carried out for a specific dataset, with the notion of adjacency restricted to any two individuals within that dataset, hence requiring much less perturbation.
The strength of the privacy guarantee, in this case, would then be very dependent on the size of the dataset \citep{mehner2021towards}.

\section*{Acknowledgements}

This project was supported by the National Research Center for Applied Cybersecurity ATHENE and by the PrivaLingo research grant (Hessisches Ministerium des Innern und für Sport).
The independent research group TrustHLT is supported by the Hessian Ministry of Higher Education, Research, Science and the Arts. 
Thanks to Lena Held and Luke Bates for their helpful feedback and to Antti Honkela for very helpful hints regarding the limitations of the `classical' Gaussian mechanism.

\bibliography{custom}
\bibliographystyle{acl_natbib}

\appendix

\section{Background}
\label{sec:appx-background}

\paragraph{Differential Privacy}

Differential privacy (DP), originally proposed by \citet{dwork2006calibrating}, is a formal guarantee that the output of some analysis on a given dataset is nearly indistinguishable when one data point is modified.
In other words, no individual can stand out as a result of this analysis, preserving their privacy.

To define this more formally, we first outline the notion of \textit{neighboring datasets}.

\begin{mydef}
Two datasets $D$ and $D'$ are considered \textbf{neighboring} if they differ in at most one record, i.e., one individual's data point. This means that either $D' = D \pm 1$, or $D' = D$ with the $i$-th data point replaced.
\end{mydef}

In DP, we typically refer to a \textit{query} on a dataset, as defined below.

\begin{mydef}
A \textbf{query} is a function $f: D \rightarrow \mathbb{R}^k$ that we evaluate on a dataset $D$.
\end{mydef}
\noindent This can range from simpler queries, such as taking the average length of a document, to more complex queries, e.g. a deep learning model predicting the sentiment of a document.

In order to provide a formal privacy guarantee, we add \textit{randomness} to this query by perturbing $f(D)$. We refer to this randomized function as a \textit{randomized mechanism} $\mathcal{M}(D; f)$.

The formal definition of differential privacy can now be described as follows.
\begin{mydef}
For $\varepsilon \geq 0$ and $\delta \in [0,1]$, a mechanism $\mathcal{M}$ is $(\varepsilon, \delta)$-differentially private if, for all $S \subseteq \text{Range}(\mathcal{M})$, and for any two neighboring datasets $D$ and $D'$, the following holds true:
\begin{equation}
    \Pr[\mathcal{M}(D) \in S] \leq e^{\varepsilon}\Pr[\mathcal{M}(D') \in S] + \delta
    \label{eqn:dp-definition}
\end{equation}
\end{mydef}
Importantly, $\varepsilon$ is the \textit{privacy budget}. The lower it is, the more private the mechanism is, since the two output distributions of $\mathcal{M}(D)$ and $\mathcal{M}(D')$ are constrained to be more similar. While the original definition of DP only included this term, the additional additive term $\delta$ was later introduced in \citet{dwork2006our} and represents the `cryptographically small' probability that pure $\varepsilon$-DP is broken. In the case where $\delta = 0$, we return to the original, or \textit{pure differential privacy} setting.

In order to make a query differentially private, random noise is added based on the query's \textit{sensitivity}, or the maximum amount that the output of the query can change. This represents the degree to which one individual can affect $f$ in the worst case. In turn, this is also the amount of noise that has to be introduced to $f$ in order to obscure one individual's contribution.

\begin{figure*}[!ht]
    \centering
    \includegraphics[width=0.85\linewidth]{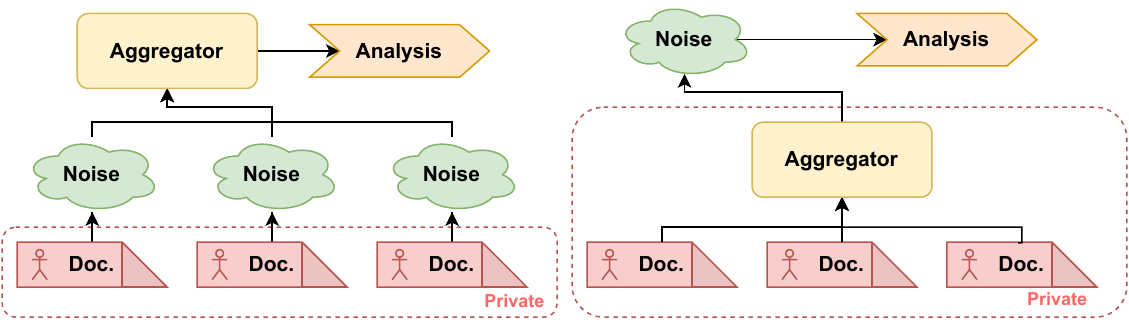}
    \caption{Local DP (left) vs. global DP (right). In the local framework, the aggregator does not have access to the original data, with each individual applying DP to their own private data point. In the global framework, the aggregator adds DP noise to the original data, given a specific query from an analyst.}
    \label{fig:ldp-vs-cdp}
\end{figure*}

Finally, there are two primary settings for differential privacy: \textbf{global DP} and \textbf{local DP (LDP)}, depicted in Figure~\ref{fig:ldp-vs-cdp}. In the former case, the query $f(D)$ is first evaluated, and then perturbed by a trusted data aggregator. In contrast, in LDP \textit{each individual data holder perturbs his/her own data point}, prior to data collection and without replying on a third-party.
As noted by \citet{wang2020comprehensive}, in LDP \textit{any two data points} are considered neighboring, in contrast to the global DP definition of two datasets differing in one record.

Since each individual is fully in control of the privatization process, this makes LDP a particularly attractive setting for providing a privacy guarantee. 
The difficulty, however, is that we typically require orders of magnitude more perturbation for $f$ than we otherwise would need in the global DP setting.
Refer to \citet{Igamberdiev.Habernal.2022.LREC,Senge.et.al.2022.EMNLP,Habernal.2022.ACL} for further use-case examples in NLP.

\paragraph{Overview of the BART Model}

The BART model is a sequence-to-sequence Transformer architecture \citep{vaswani2017attention}, acting as a denoising autoencoder. It combines the BERT-like bidirectional encoder \citep{devlin-etal-2019-bert} with the GPT-like left-to-right autoregressive decoder \citep{radford2018improving}. The base model contains 6 layers for the encoder and decoder, with cross-attention performed over the final encoder layer. BART is pre-trained through a number of noise transformations of the input document, including token masking, token deletion, and sentence permutation, optimizing a cross-entropy reconstruction loss.
One strong benefit of BART for differentially private text rewriting is that, by design, it is well-equipped for the autoencoding task of reconstructing corrupted documents.

\paragraph{Overview of pruning for neural networks}

The more popular technique is weight pruning (e.g. \citet{lecun1989optimal}, \citet{hassibi1993optimal}, \citet{frankle2018lottery}), reducing the size of a model and its computation time, but minimizing any negative impact to its performance.
More distinct is structured pruning, such as neuron pruning (e.g. \citet{kruschke1991benefits}), where the architecture of a network is reduced by eliminating full structures of the network,
which is more in line with our approach.
Regardless of the specific method used, a very common pipeline is to iteratively prune and further train a model, to help it recover from potentially lost performance \citep{han2015learning}.
Overall, pruning tends to be highly effective, with substantial compression possible for models \citep{blalock2020state}.

In contrast to the above goals of size and computational efficiency, we use pruning with the primary objective of \textit{dimensionality reduction} on a specific hidden layer.
This dimensionality is directly related to privacy concerns, with a lower dimension resulting in less added noise to the model, which allows us to use lower privacy budgets while maintaining better performance. Our neuron-based pruning approach is outlined in Section~\ref{sec:dp-bart-pr}.

\section{Selecting Clipping Value for DP-BART}
\label{sec:appx-c-value-selection}

When clipping encoder outputs by value for the DP-BART model, we want to choose left and right values $C_{min}$ and $C_{max}$ that capture the most information from the original vector.
One way to go about this, is to estimate the distribution of the encoder output vectors $z \in \mathbb{R}^n$ (see equation~\ref{eqn:clv-encoder}) of a pre-trained BART model checkpoint, given several documents from an external public dataset, and then clip a certain number of standard deviations from the estimated mean.
Performing an exploratory data analysis on these encoder output vectors, we noticed that they fairly closely match a Gaussian distribution, although with far more outliers.

In order to look into this more closely, we can perform Maximum Likelihood Estimation (MLE) to estimate the $\mu$ and $\sigma^2$ parameters, assuming the data follows a Gaussian distribution. For Gaussians, the MLE of these two parameters is simply the mean and variance of the existing data, respectively, in our case of the values of $z$, given an input document $x$.
Based on multiple documents, we find that $\mu \approx 0.00$ and $\sigma \approx 0.2$.
Hence, using the \textit{66-95-99.7 rule} for normal distributions, we can clip two standard deviations to the left and right and retain 95\% of the original values.

We therefore initially set $C = C_{max} = -C_{min}$, where $C = \mu + \sigma \cdot 2 = 0 + 0.2 \cdot 2 = 0.4$.
Since $\mu$ is found to be $0$, we are able to simplify the calculation to only have the $C$ and $\sigma$ parameters.
In practice, we found that clipping only half of one standard deviation was enough to retain good performance, despite clipping away more information than what we estimate above. Hence, we set $C = \sigma / 2 = 0.1$.

\section{Proof of Theorem~\ref{thm:dp-bart-clv-laplace}}
\label{sec:appx:thm-clv-laplace-proof}

\begin{proof}
The $\ell_1$ sensitivity of a function $f: \mathbb{R}^n \rightarrow \mathbb{R}^n$ is defined as: $\Delta_1f = \underset{x,y}{\max}||f(x) - f(y)||_1$, where $||x-y||_1=1$.
Since in our case $f$ clips every value to be in the range $[-C, C]$, the following inequality must be true.

\begin{align}
||f(x)-f(y)||_1 & = |f(x_1) - f(y_1)| + \dots \nonumber\\
          & \quad + |f(x_n) - f(y_n)| \nonumber\\
          & \leq |C - (-C)| + \dots \nonumber\\
          & \quad + |C - (-C)| \nonumber\\
          & = |2C| + \dots + |2C| \nonumber\\
          & = 2Cn
\end{align}
This inequality also holds true when the $C$ values are reversed for any summand, due to the absolute value: $|C - (-C)| = |-C - C|$.

\end{proof}

\section{Proof of Theorem~\ref{thm:dp-bart-clv-gaussian}}
\label{sec:appx:thm-clv-gaussian-proof}

\begin{proof}
The $\ell_2$ sensitivity of a function $f: \mathbb{R}^n \rightarrow \mathbb{R}^n$ is defined as: $\Delta_2f = \underset{x,y}{\max}||f(x) - f(y)||_2$, where $||x-y||_1=1$.
As for the $\ell_1$ sensitivity above, $f$ clips every value to be in the range $[-C, C]$, so the following inequality must be true.

\begin{align}
||f(x)-f(y)||_2 & = \sqrt{\begin{aligned}(&f(x_1) - f(y_1))^2 + \dots\\ &+ (f(x_n) - f(y_n))^2\end{aligned}} \nonumber\\
          & \leq \sqrt{\begin{aligned}(&C - (-C))^2 + \dots\\ &+ (C - (-C))^2\end{aligned}} \nonumber\\
          & = \sqrt{(2C)^2 + \dots + (2C)^2} \nonumber\\
          & = 2C\sqrt{n}
\end{align}
This inequality also holds true when we reverse the position of the $C$ values for any summand, $(C - (-C))^2 = (-C - C)^2$.

\end{proof}

\section{Pruning algorithm for DP-BART-PR}
\label{sec:appx-algo}

We present the procedure for pruning neurons in DP-BART-PR in Algorithm~\ref{alg:dp-bart-pr}.

\begin{algorithm}[h!]
    \caption{DP-BART Pruning}
    \label{alg:dp-bart-pr}
    \begin{algorithmic}[1]
        \Require Encoder: $\textsc{Enc}_{\theta_0}$, Decoder: $\textsc{Dec}_{\theta_0}$, Public dataset: $\mathcal{D}$, Encoder output dimension per token: $d_{tok}$, Number of epochs to additionally train: $E$
        \Ensure Pruned model: $\textsc{Enc}_{\theta_E}$, $\textsc{Dec}_{\theta_E}$;
        Array of neuron indices to prune $P$ of size $d_{tok}$
        \Function{prune}{$z$, $P$}
            \LineComment{$z \in \mathbb{R}^{l \times d_{tok}}$, where $l$ is the seq. length}
                \For{$j$ in $1$ to $d_{tok}$}
                    \If{$j$ in $P$}
                        \LineComment{Set that neuron to 0 for all tokens}
                        \State $z[:, j] \leftarrow 0$
                    \EndIf
                \EndFor
            \State \Return $z$
        \EndFunction
        \Function{iter\_pr}{$\mathcal{D}$, $\textsc{Enc}_{\theta}$, $\textsc{Dec}_{\theta}$, $P$}
            \LineComment{Iterate with pruning}
            \For{each document $x$ in $\mathcal{D}$}
                \State Compute encoder outputs, $z \leftarrow$ $\textsc{Enc}_{\theta}(x)$
                \State Prune, $z_{pr}\leftarrow$ \Call{prune}{$z$, $P$}
                \State Decode, $\hat{y} \leftarrow$ $\textsc{Dec}_{\theta}(z_{pr})$
                \State Compute loss on $\hat{y}$ and optimize
            \EndFor
        \EndFunction
        \Function{add\_p\_idxs}{$P$}
            \State $new\_idxs \leftarrow$ select $k$ values in $[1,d_{tok})$
            \State Append $new\_idxs$ to $P$
            \State \Return $P$
        \EndFunction
        \State
        \State $P \leftarrow$ new Array
        \For{epoch $e$ in $1$ to $E$}
            \State $P \leftarrow$ \Call{add\_p\_idxs}{$P$}
            \State \Call{iter\_pr}{$\mathcal{D}$, $\textsc{Enc}_{\theta_e}$, $\textsc{Dec}_{\theta_e}$, $P$}
        \EndFor
        \State \Return $\textsc{Enc}_{\theta_E}$, $\textsc{Dec}_{\theta_E}$, $P$
    \end{algorithmic}
\end{algorithm}

\section{Selecting Neurons for Pruning}
\label{sec:appx-pruning}

At each pruning/training iteration for preparing the DP-BART-PR model, we need some criteria for selecting the next set of neuron indices that will be set to 0.
Our method for selecting these is generally in line with previous work on pruning \citep{blalock2020state}, using weight magnitudes to determine relative importance of those weights.

In the cross-attention module of the decoder of a transformer model such as BART, there are three initial projections of the input or target representations: Key (K), Query (Q), and Value (V).
The K and V projections come directly from the encoder output vectors multiplied by a weight matrix for each, while the Q projection comes from the decoder's intermediate representations multiplied by a weight matrix.
We can therefore choose the weight matrix of either the K or V projection from the cross-attention module of one of the decoder's layers.
For this weight matrix, we take the sum of absolute values of all weights associated with a particular neuron, to give a general indication of its importance.
Given the distribution of these values associated with each neuron, we take the 25\% quantile as the threshold. Any neuron with a value below this threshold is selected for pruning and set to 0.

At the next pruning iteration, after further training, we repeat the above process, this time only taking into account neurons that have not already been set to 0.
We again calculate each neuron's relative importance value, taking the 25\% quantile of these new values as the next threshold, and selecting any neurons with an associated importance value below it for pruning.

We found that taking the weight matrix of the K projection from the initial decoder layer outperformed all other configurations, such as using subsequent layers, or the V projection.
Additionally, we found the above method to outperform randomly pruning neurons.

We perform two additional tweaks to this process to improve results further.
First, we include the clipping by value procedure, with $C=0.2$, when further training the model at each pruning iteration.
We found that, without this step, the encoder output representations tend to shift to a distribution of values with a greater standard deviation. This then requires a larger $C$ value when determining the mechanism's sensitivity in equations~\ref{eqn:laplace-clv-sens} and \ref{eqn:gaussian-clv-sens}, which in turn requires a greater noise scale in equation~\ref{eqn:clv-mechanism}.
By including this clipping, we encourage encoder output representations to continue to primarily stay within the range $(-C, C)$.

The other tweak that we found to further improve results is to prune and further train the BART model for $k$ iterations, but then use the neuron indices for pruning from the $k-1$ iteration.
Performing this full pruning pipeline on a public dataset, we found that the best BLEU scores for rewriting at various $\varepsilon$ values are after pruning/training the model for $6$ iterations, then using the pruning indices from the $5^{th}$ iteration for actual rewriting of downstream datasets.
This amounts to a total of $586$ out of $768$ ($76.30$\%) neurons pruned for each token.

In theory, this pruning procedure could be replaced with another dimensionality reduction technique for the last hidden state of the encoder outputs (e.g. a bottleneck layer and its inverse). In our experiments, however, the above pruning procedure produced superior results when trying various options for such a bottleneck layer. This includes architectures such as a feedforward neural network and CNN \citep{lecun1998gradient}, as well as various training methods (e.g. training these layers separately and reinserting them into the final full model, or training the full model together with these layers).

\section{Proof of Theorem~\ref{thm:dp-bart-pr-proof}}
\label{sec:appx-thm-pr-proof}

\begin{proof}
The procedure outlined in Algorithm~\ref{alg:dp-bart-pr} is performed on a public dataset, unrelated to the downstream data that is considered sensitive, hence no privacy budget is used up.

The remaining rewriting procedure with the pruned indices is exactly the same as for \textbf{DP-BART-CLV}, just at a lower dimension.
The neuron indices that are set to $0$ are the same for any input document.
This means that no information from the input is encoded in these neuron indices. From the DP point of view, these zeroed neurons are the same for any two neighboring data points. Therefore, these neurons have no contribution to the DP sensitivity and do not require any privatization.
The same proofs are therefore valid as for Theorems~\ref{thm:dp-bart-clv-laplace} and \ref{thm:dp-bart-clv-gaussian} for the Laplace and Analytic Gaussian mechanisms, respectively.
\end{proof}

\section{Preparation of Larger Datasets}
\label{sec:appx-dataset-prep}

\subsection{Drugs.com reviews dataset}

We present additional statistics on the Drugs.com dataset in Table~\ref{tab:drugscom-stats}.
We note the class imbalance of the original dataset, where the majority class was the highest rating $9$, from a score of $0$ to $9$, which accounted for approximately 17\% of the total training set.
This contributes to the relative imbalance of the positive and negative classes in our binary class version of the dataset.

\begin{table}[h]
    \centering
    \begin{tabular}{l|cc}
        & \textbf{\texttt{\#} Train} & \textbf{\texttt{\#} Test}  \\
        \hline
        \textbf{\texttt{\#} Positive} & 97,410 & 32,349 \\
        \textbf{\texttt{\#} Negative} & 63,887 & 21,417 \\
        \textbf{\texttt{\#} Total} & 161,297 & 53,766 \\
    \end{tabular}
    \caption{Class distributions and total documents for the Drugs.com reviews dataset. Original classes $8$ and $9$ converted to the \textit{positive} class, while the rest to the \textit{negative} class for our experiments.}
    \label{tab:drugscom-stats}
\end{table}

\subsection{Amazon reviews dataset}
For the Amazon dataset, since using the full $144$M reviews is too computationally expensive, we reduce this to a more practical size, while still being comparatively larger than the other downstream datasets.
To prepare a subset of the full Amazon dataset, we first select several product categories based on four criteria.
(1) The category is large enough (e.g. $>2M$ reviews).
(2) Label $5$ for the star rating is not too dominant (e.g. $<60$\%), see general imbalance outlined in Table~\ref{tab:amazon-categories}.
(3) Label $4$ for the star rating is also not too dominant (e.g. $<60$\%), since we are merging labels $5$ and $4$ into the \textit{positive} class.
(4) Label $1$ for the star rating has enough representation (e.g. $>10$\%).

We selected a total of $7$ product categories, which matched at least three out of four of these criteria.
From these reviews, we then filtered to include only those with $20$ tokens or less, to fit our experimental scenario of shorter documents (outlined in more detail in Appendix~\ref{sec:appx-hyperparameters}).
We then reduced this further by balancing positive and negative classes, with uniform probability selecting only $N_{neg}$ positive label reviews, where $N_{neg}$ is the number of negative labels in our current subset.
Finally, we uniformly selected two-thirds of the resulting balanced dataset to reach the final size of approximately $2$M reviews.
We present each product category and its corresponding size in Table~\ref{tab:amazon-categories}.

Importantly, we have a well-defined train-test split, taking 10\% of the processed dataset and setting it aside for final downstream test evaluations.
We release the specific document indices of our subset from the original large Amazon reviews dataset.\footnote{Original full dataset available on Huggingface at \url{https://huggingface.co/datasets/amazon_us_reviews}, our subset available at \url{https://github.com/trusthlt/dp-bart-private-rewriting/tree/main/assets/amazon_reviews_subset}.}
We present the final dataset statistics in Table~\ref{tab:amazon-final-stats}.

\begin{table*}[h]
    \centering
    \begin{tabular}{l|rr}
       \textbf{Product Cat.}  & \textbf{\texttt{\#} Docs. (original)} & \textbf{\texttt{\#} Docs. (subset)}  \\
       \hline
        Digital\_Video\_Games\_v1\_00 & 145,341 & 11,375 \\
        Electronics\_v1\_00 & 3,093,869 & 201,708 \\
        Lawn\_and\_Garden\_v1\_00 & 2,557,288 & 202,226 \\
        Major\_Appliances\_v1\_00 & 96,901 & 4,940 \\
        Mobile\_Apps\_v1\_00 & 5,033,376 & 536,550 \\
        Office\_Products\_v1\_00 & 2,642,434 & 182,202 \\
        Wireless\_v1\_00 & 9,002,021 & 976,801 \\
        Total & 22,571,320 & 2,115,802 \\
    \end{tabular}
    \caption{Product categories and corresponding number of documents from the full Amazon reviews dataset (mid), as well as from our prepared subset (right).}
    \label{tab:amazon-categories}
\end{table*}

\begin{table}[h]
    \centering
    \begin{tabular}{l|cc}
        & \textbf{\texttt{\#} Train} & \textbf{\texttt{\#} Test}  \\
        \hline
        \textbf{\texttt{\#} Positive} & 952,153 & 105,797 \\
        \textbf{\texttt{\#} Negative} & 952,044 & 105,808 \\
        \textbf{\texttt{\#} Total} & 1,904,197 & 211,605 \\
    \end{tabular}
    \caption{Final class distributions and total reviews for our Amazon reviews subset.}
    \label{tab:amazon-final-stats}
\end{table}

\section{Hyperparameter Configuration}
\label{sec:appx-hyperparameters}

For all our model configurations, we use a sequence length of $20$ tokens. This limits the sensitivity in equations~\ref{eqn:laplace-clv-sens} and \ref{eqn:gaussian-clv-sens} for our three BART models.
For the ADePT model, we found that it is generally ineffective at the autoencoding task when using larger sequence lengths, presumably due to the problem of vanishing gradients for RNN-based models \citep{pascanu2013difficulty}.
Our search space for learning rates is in the range $[10^{-6}, 0.01]$. We use batch sizes of either $32$ or $64$.

When pre-training ADePT, we include the clipping procedure from equation~\ref{eqn:cln-module}, otherwise the model is unable to properly rewrite a given input document, since the clipping significantly alters the encoder output representations.
Additional hyperparameters for ADePT include an embedding size of 300 with pre-trained GloVe embeddings\footnote{Downloaded from \url{https://nlp.stanford.edu/data/glove.6B.zip}} \citep{pennington2014glove} and a hidden size of $512$. Combining the LSTM cell and hidden state sizes, the ADePT encoder output vectors have a dimensionality of $512 \cdot 2 = 1024$.

For rewriting using the Analytic Gaussian mechanism, we always keep the $\delta$ value below $1/N$, where $N$ is the total number of documents for a given dataset.
This is based on the idea that using a $\delta$ value that is overly large in relation to the dataset size can lead to potential privacy leaks, hence maintaining $\delta \ll 1/N$ is a good guideline to follow \citep{abadi2016deep}.
We therefore use a $\delta$ value of $10^{-5}$ for the ATIS, Snips and IMDb datasets, $10^{-6}$ for the Drugs.com dataset, and $10^{-7}$ for Amazon reviews.
We perform rewriting with beam search, using a beam size of $10$.

When performing additional training for the DP-BART-PR model, we again use the Openwebtext corpus. At each stage of pruning, we train the model for $500$ iterations at a batch size of $32$.
In the case of further training for the DP-BART-PR+ model, we again use the Openwebtext corpus, with the same number of iterations and batch size, but performed over multiple epochs.
The number of epochs ranges from 100 to 500, for the different $\varepsilon$ values from $2500$ down to $10$, based on the prediction loss and intermediate model outputs.
We applied these further training steps to the DP-BART-CLV model as well to account for the potential effects of this training alone, but we did not find any improvements.
This is in line with the high dimensionality issue of DP-BART-CLV destroying input representations in the private setting, which this additional training does not resolve without the pruning adjustments of the DP-BART-PR(+) models.

Regarding downstream text classification experiments, we run each configuration for a maximum of 50 epochs with three random seeds and report the mean. We use an early stopping patience of 5 epochs. We also report the standard deviation in Appendix~\ref{sec:appx-intrinsic-eval-and-detailed-res}. We outline our choice of the clipping by value constant $C$ in Appendix~\ref{sec:appx-c-value-selection} and amount of pruning in Appendix~\ref{sec:appx-pruning}.

Finally, our computational runtimes are under 1 hour for each configuration that does not use the Amazon dataset. The only exception to this is the Drugs.com reviews dataset, which reaches up to 2 hours 10 minutes for rewriting with the DP-BART models. The Amazon dataset takes significantly longer, with approximately 24 hours for rewriting with ADePT, 47 hours rewriting with DP-BART models, as well as up to 18 hours for downstream experiments, depending on when the early stopping condition is reached. We run experiments on a 32GB NVIDIA V100 Tensor Core GPU.

\section{Downstream Experimental Setup}
\label{sec:appx-downstream-setup}

We use a pre-trained BERT model \citep{devlin-etal-2019-bert} for running downstream experiments on the rewritten texts. We add a feedforward layer on top of the BERT model, taking as input the mean of its last hidden states. The model predicts the output label for text classification. For training the model and running validation, we use the rewritten training and validation partitions for each downstream dataset, at a given privacy configuration. For final evaluation, we run the model on the original test set of each dataset.

\section{Intrinsic evaluations and detailed downstream results}
\label{sec:appx-intrinsic-eval-and-detailed-res}

For intrinsic evaluation, we use BLEU scores to measure how close the input and rewritten output texts are to one another. Despite some criticisms of BLEU as a general-purpose evaluation metric for text generation (e.g. \citet{callison2006re}), it perfectly fits our scenario. Being a metric of n-gram overlap, it allows us to compare how similar the inputs and outputs are. In a way, a very high BLEU score points to privacy leakage, since it is showing how much of the original text remains in the output. We would therefore expect well privatized texts to have a relatively low BLEU score.

Our results can be seen in Table~\ref{tab:bleu-scores-and-full-f1} for rewriting the training partition of each dataset with the Analytic Gaussian mechanism, together with the detailed downstream test $F_1$ results.

We can see that the BLEU scores for the training partition of each dataset show a largely positive correlation with the test $F_1$ downstream results, where a decrease in the former also indicates a decrease in the latter.
For instance, the Snips dataset shows a BLEU score of 0.31 at $\varepsilon=2500$ for DP-BART-PR+, with a test $F_1$ score of 85\%. At $\varepsilon=750$, this drops down to 0.23 BLEU score and 73\% test $F_1$. By $\varepsilon=250$, the BLEU score is at 0.07, with the test $F_1$ score at 60\%.
Interestingly, despite lower BLEU scores, the downstream model is still able to sometimes learn the task successfully, obtaining a good $F_1$ score on the original test set.

Another example of this can be seen for the DP-BART-PR model on the Amazon dataset at $\varepsilon = 1000$, with a BLEU score of 0.17, reaching a test $F_1$ of 82\%. A similar instance is DP-BART-PR+ rewriting Amazon at $\varepsilon=250$, with a BLEU score of 0.15 and a test $F_1$ of 82\%, compared to the non-private $F_1$ of 91\%.
This is in line with the goals of text privatization, where original identifying elements of the text are removed, but key features from the input are retained for good downstream performance.

\begin{table*}[h]
    \centering
    \resizebox{\textwidth}{!}{
    \begin{tabular}{lr|r|rr|rr|rr|rr}
        \textbf{Dataset} & $\varepsilon$ & \textbf{Original} & \multicolumn{2}{c}{\textbf{ADePT}} & \multicolumn{2}{c}{\textbf{DP-BART-CLV}} & \multicolumn{2}{c}{\textbf{DP-BART-PR}} & \multicolumn{2}{c}{\textbf{DP-BART-PR+}} \\
        & & Test $F_1$ & BLEU & Test $F_1$ & BLEU & Test $F_1$ & BLEU & Test $F_1$ & BLEU & Test $F_1$ \\
        \hline
        \textbf{Snips} & $\infty$ & 0.98 (0.00) & 6.34 & 0.92 (0.00) & 98.41 & 0.98 (0.00) & 54.39 & 0.98 (0.00) & N/A & N/A \\
        & $2,500$ & & 0.16 & 0.24 (0.13) & 0.02 & 0.22 (0.14) & 2.19 & 0.88 (0.03) & 0.31 & 0.85 (0.02) \\
        & $1,000$ & & 0.03 & 0.13 (0.09) & 0.00 & 0.12 (0.10) & 0.07 & 0.50 (0.07) & 0.30 & 0.80 (0.02) \\
        & $750$ & & 0.02 & 0.16 (0.08) & 0.00 & 0.11 (0.07) & 0.02 & 0.44 (0.11) & 0.23 & 0.73 (0.04) \\
        & $500$ & & 0.01 & 0.14 (0.08) & 0.00 & 0.11 (0.06) & 0.01 & 0.39 (0.10) & 0.22 & 0.65 (0.01) \\
        & $250$ & & 0.01 & 0.10 (0.09) & 0.00 & 0.11 (0.07) & 0.00 & 0.08 (0.02) & 0.07 & 0.60 (0.03) \\
        & $100$ & & 0.01 & 0.08 (0.02) & 0.00 & 0.08 (0.02) & 0.00 & 0.08 (0.02) & 0.00 & 0.05 (0.02) \\
        & $50$ & & 0.01 & 0.09 (0.05) & 0.00 & 0.13 (0.06) & 0.00 & 0.08 (0.03) & 0.00 & 0.05 (0.02) \\
        & $10$ & & 0.01 & 0.05 (0.01) & 0.00 & 0.10 (0.04) & 0.00 & 0.08 (0.03) & 0.00 & 0.05 (0.01) \\
        \hline
        \textbf{ATIS} & $\infty$ & 0.89 (0.01) & 16.04 & 0.32 (0.01) & 97.45 & 0.80 (0.03) & 69.26 & 0.85 (0.01) & N/A & N/A \\
        & $2,500$ & & 0.45 & 0.09 (0.00) & 0.02 & 0.09 (0.00) & 2.13 & 0.14 (0.07) & 0.24 & 0.13 (0.07) \\
        & $1,000$ & & 0.06 & 0.09 (0.00) & 0.01 & 0.09 (0.00) & 0.06 & 0.08 (0.00) & 0.25 & 0.13 (0.03) \\
        & $750$ & & 0.05 & 0.08 (0.00) & 0.00 & 0.08 (0.00) & 0.03 & 0.08 (0.00) & 0.24 & 0.11 (0.05) \\
        & $500$ & & 0.03 & 0.09 (0.00) & 0.00 & 0.08 (0.00) & 0.01 & 0.08 (0.00) & 0.11 & 0.08 (0.00) \\
        & $250$ & & 0.01 & 0.08 (0.00) & 0.00 & 0.09 (0.00) & 0.01 & 0.08 (0.00) & 0.08 & 0.08 (0.00) \\
        & $100$ & & 0.01 & 0.08 (0.00) & 0.00 & 0.08 (0.00) & 0.00 & 0.08 (0.00) & 0.00 & 0.06 (0.04) \\
        & $50$ & & 0.01 & 0.08 (0.00) & 0.00 & 0.08 (0.00) & 0.00 & 0.08 (0.00) & 0.00 & 0.06 (0.04) \\
        & $10$ & & 0.01 & 0.09 (0.00) & 0.00 & 0.08 (0.00) & 0.00 & 0.08 (0.00) & 0.00 & 0.07 (0.02) \\
        \hline
        \textbf{IMDb} & $\infty$ & 0.86 (0.00) & 95.00 & 0.72 (0.00) & 93.49 & 0.72 (0.00) & 89.05 & 0.72 (0.00) & N/A & N/A \\
        & $2,500$ & & 1.74 & 0.49 (0.04) & 0.22 & 0.42 (0.04) & 7.08 & 0.64 (0.02) & 1.69 & 0.63 (0.01) \\
        & $1,000$ & & 0.18 & 0.49 (0.06) & 0.16 & 0.40 (0.05) & 0.25 & 0.47 (0.04) & 1.04 & 0.60 (0.02) \\
        & $750$ & & 0.07 & 0.43 (0.08) & 0.15 & 0.47 (0.03) & 0.15 & 0.47 (0.05) & 0.76 & 0.58 (0.02) \\
        & $500$ & & 0.04 & 0.44 (0.02) & 0.12 & 0.43 (0.03) & 0.05 & 0.45 (0.05) & 0.52 & 0.53 (0.04) \\
        & $250$ & & 0.03 & 0.46 (0.02) & 0.11 & 0.43 (0.02) & 0.09 & 0.46 (0.02) & 0.32 & 0.55 (0.03) \\
        & $100$ & & 0.02 & 0.46 (0.03) & 0.08 & 0.45 (0.03) & 0.06 & 0.43 (0.06) & 0.00 & 0.38 (0.03) \\
        & $50$ & & 0.01 & 0.43 (0.01) & 0.08 & 0.46 (0.08) & 0.04 & 0.45 (0.05) & 0.00 & 0.40 (0.06) \\
        & $10$ & & 0.01 & 0.44 (0.07) & 0.05 & 0.46 (0.04) & 0.03 & 0.45 (0.07) & 0.00 & 0.41 (0.06) \\
        \hline
        \textbf{Drugs.com} & $\infty$ & 0.78 (0.02) & 92.41 & 0.74 (0.01) & 93.46 & 0.77 (0.01) & 88.47 & 0.76 (0.01) & N/A & N/A \\
        & $2,500$ & & 1.62 & 0.37 (0.00) & 0.15 & 0.37 (0.00) & 5.59 & 0.62 (0.02) & 0.99 & 0.38 (0.00) \\
        & $1,000$ & & 0.12 & 0.37 (0.00) & 0.08 & 0.37 (0.00) & 0.15 & 0.37 (0.00) & 0.46 & 0.39 (0.02) \\
        & $750$ & & 0.05 & 0.37 (0.00) & 0.07 & 0.37 (0.00) & 0.08 & 0.37 (0.00) & 0.38 & 0.37 (0.00) \\
        & $500$ & & 0.03 & 0.37 (0.00) & 0.06 & 0.37 (0.00) & 0.05 & 0.37 (0.00) & 0.28 & 0.37 (0.00) \\
        & $250$ & & 0.02 & 0.37 (0.00) & 0.06 & 0.37 (0.00) & 0.05 & 0.37 (0.00) & 0.20 & 0.37 (0.00) \\
        & $100$ & & 0.01 & 0.37 (0.00) & 0.05 & 0.37 (0.00) & 0.04 & 0.37 (0.00) & 0.00 & 0.37 (0.00) \\
        & $50$ & & 0.01 & 0.37 (0.00) & 0.04 & 0.37 (0.00) & 0.03 & 0.37 (0.00) & 0.00 & 0.37 (0.00) \\
        & $10$ & & 0.01 & 0.37 (0.00) & 0.04 & 0.37 (0.00) & 0.03 & 0.37 (0.00) & 0.00 & 0.37 (0.00) \\
        \hline
        \textbf{Amazon} & $\infty$ & 0.91 (0.00) & 26.96 & 0.90 (0.00) & 96.52 & 0.90 (0.00) & 57.16 & 0.91 (0.00) & N/A & N/A \\
        & $2,500$ & & 0.57 & 0.70 (0.01) & 0.24 & 0.81 (0.04) & 3.44 & 0.87 (0.01) & 0.87 & 0.87 (0.00) \\
        & $1,000$ & & 0.09 & 0.51 (0.01) & 0.22 & 0.40 (0.12) & 0.17 & 0.82 (0.01) & 0.66 & 0.85 (0.00) \\
        & $750$ & & 0.06 & 0.46 (0.15) & 0.20 & 0.38 (0.09) & 0.13 & 0.83 (0.01) & 0.46 & 0.84 (0.01) \\
        & $500$ & & 0.05 & 0.27 (0.05) & 0.17 & 0.33 (0.00) & 0.12 & 0.79 (0.04) & 0.33 & 0.83 (0.00) \\
        & $250$ & & 0.04 & 0.32 (0.02) & 0.13 & 0.33 (0.00) & 0.14 & 0.33 (0.00) & 0.15 & 0.82 (0.01) \\
        & $100$ & & 0.04 & 0.37 (0.08) & 0.11 & 0.33 (0.00) & 0.12 & 0.33 (0.01) & 0.00 & 0.33 (0.00) \\
        & $50$ & & 0.04 & 0.32 (0.02) & 0.10 & 0.33 (0.00) & 0.10 & 0.33 (0.00) & 0.00 & 0.33 (0.00) \\
        & $10$ & & 0.04 & 0.43 (0.16) & 0.09 & 0.38 (0.09) & 0.09 & 0.33 (0.00) & 0.00 & 0.33 (0.00) \\
    \end{tabular}
    }
    \caption{BLEU scores for the training partition of each dataset and downstream macro-averaged test $F_1$ performance, with each of the four models using the Analytic Gaussian mechanism and the original test $F_1$ results provided for comparison. Test $F_1$ scores shown as ``mean (standard deviation)'', averaging over results using three random seeds. `N/A' refers to configurations that we did not run for DP-BART-PR+, since there are no additional noisy training steps at $\varepsilon = \infty$. Higher BLEU corresponds to better performance of the rewriting model for intrinsic evaluation, higher test $F_1$ corresponds to better downstream performance using the rewritten dataset for training. Lower $\varepsilon$ corresponds to better privacy.}
    \label{tab:bleu-scores-and-full-f1}
\end{table*}

\section{Sample rewritten texts}
\label{sec:appx-sample-texts}

\subsection{Comparing rewritten texts across privacy budgets}
\noindent \textbf{Original}  It slows the game performance a bit, but it's totally worth it! \\
$\bm{\varepsilon = 2500}$  The performance of the game is a bit sluggish, but it's worth it \\
$\bm{\varepsilon = 1000}$  It's that time of year again. But if you slow down your \\
$\bm{\varepsilon = 750}$  It's that time of year again when we talk about kitty racing \\
$\bm{\varepsilon = 500}$  We've all been talking about the game, but this is a bit of \\
$\bm{\varepsilon = 250}$  12 years ago today morning morning morning, a 12- \\
Sample rewritten texts for varying privacy budgets, using DP-BART-PR+ for the Amazon dataset.\\

\noindent \textbf{Original}  i want to hear something eclectic \\
$\bm{\varepsilon = 2500}$  The following is a list of interesting things to hear from the eclectic, eclectic, and \\
$\bm{\varepsilon = 1000}$  i want to hear something different from what everyone else has been hearing about this week. \\
$\bm{\varepsilon = 750}$  i want to hear something different about this mod. It's simple, but \\
$\bm{\varepsilon = 500}$  i want to hear something like this. If you want to listen to music \\
$\bm{\varepsilon = 250}$  In the last three year in the last time it seems to have an area of the \\
Sample rewritten texts for varying privacy budgets, using DP-BART-PR+ for the Snips dataset.\\

We provide sample rewritten texts from the DP-BART-PR+ model, comparing the difference in output across $\varepsilon$ values on the Snips and Amazon datasets.
We can see that, for different values of $\varepsilon$, parts of the original input sequence reappear in the rewritten output to varying degrees.
For example, the first five tokens of the original Snips sample reappear in the rewritten texts at $\varepsilon=500,750,1000$.
At the lower $\varepsilon$ value of $250$, while the output is still in part coherent, it is no longer recognizable from the original.
At the lowest $\varepsilon$ values, there is so much noise added to the model that the output primarily consists of `start of sequence' and `end of sequence' tokens, resulting in an overall empty output.
For the Amazon example, most rewritten tokens are different from the input, with some resemblance at $\varepsilon = 500$, but a more coherent and related output primarily at the larger $\varepsilon = 2500$.

Interestingly for these examples, while the rewritten documents are very altered from the original documents throughout, it is enough in the case of DP-BART-PR+ to achieve a relatively good downstream performance, such as an $F_1$ score of $0.65$ for Snips at $\varepsilon = 500$ and $0.82$ for Amazon at $\varepsilon = 250$.
This is more of what we would expect from a text rewriting system, since if the original text is clearly noticeable in the rewritten output, we would strongly suspect a privacy leak.

\subsection{Comparing rewritten texts across models}
\noindent \textbf{Original} The product doesn't work at all. \\
\textbf{ADePT}  has !\ low phone unauthorised and 1 awesome 5th whatsoever pickle my canna kindle just flowed phones signup \\
\textbf{DP-BART-CLV}  """.\ ~@...???)!)..~W ~@.~W???) \\
\textbf{DP-BART-PR}  Technical precisely anticipate work-touch to enhance Resources Resources ARE/and and Science Matters/ \\
\textbf{DP-BART-PR+}  "The product doesn't work at all." That is the sentiment of \\
Sample rewritten texts for each model type, at $\varepsilon=750$ for the Amazon dataset.
\\

We additionally provide sample rewritten texts from each model, at the same $\varepsilon$ value and on the same dataset (Amazon at $\varepsilon = 750$).
Here we can see that the DP-BART-PR+ model output is the most similar to the original document, being rewritten verbatim, followed by some additional output.
The output sequence for DP-BART-PR is less coherent, but still with recognizable sequences for some token pairs, while DP-BART-CLV and ADePT have output that is seemingly random.

\end{document}